\documentclass[oneside,10pt]{article}
\usepackage{amsmath,amssymb,amsthm,amsfonts,amsbsy,latexsym}
\usepackage{dsfont}

\usepackage{graphicx}

\newcommand{\bdm}{\begin{displaymath}}
\newcommand{\edm}{\end{displaymath}}
\newcommand{\bea}{\begin{eqnarray}}
\newcommand{\eea}{\end{eqnarray}}
\newcommand{\beas}{\begin{eqnarray*}}
\newcommand{\eeas}{\end{eqnarray*}}
\newcommand{\bay}{\begin{array}{c}}
\newcommand{\eay}{\end{array}}
\newcommand{\ben}{\begin{enumerate}}
\newcommand{\een}{\end{enumerate}}
\newcommand{\be}{\begin{equation}}
\newcommand{\ee}{\end{equation}}

\newcommand{\laa}{\langle\hspace{-0.08cm}\langle}
\newcommand{\raa}{\rangle\hspace{-0.08cm}\rangle}

\newcommand{\landau}{\mbox{\begin{scriptsize}$\mathcal{O}$\end{scriptsize}}}
\newcommand{\LZ}{L^2(\mathbb{R}^3)}
\newcommand{\LZN}{L^2(\mathbb{R}^{3N})}

\newtheorem{theorem}{Theorem}[section]
\newtheorem{lemma}[theorem]{Lemma}

\newtheorem*{remark}  {Remark}
\newtheorem{definition}[theorem] {Definition}

\renewcommand{\phi}{\varphi}

\begin{document}

\title{A simple  derivation of mean field limits for quantum systems}

\author{Peter Pickl\footnote{
pickl@itp.phys.ethz.ch
Institute of Theoretical Physics, ETH H\"onggerberg, CH-8093
Z\"urich, Switzerland}}

\maketitle

\begin{abstract}

We shall present a new strategy for handling mean field limits of quantum mechanical systems. The new method is simple and
effective. It is simple, because it translates the idea behind the mean field description of a many particle quantum system
directly into a mathematical algorithm. It is effective because the strategy yields with lesser effort better results than previously achieved. As an instructional example we treat a simple model for the time dependent Hartree equation
which we derive under more general conditions than what has been considered so far. Other mean field scalings leading e.g. to the Gross-Pitaevskii equation
can also be treated \cite{pickl1,pickl2}.
\end{abstract}

\section{Introduction}

The dynamics of a quantum mechanical many body systems with interaction can sometimes be well approximated by an
effective description in which each particle moves in the mean field generated by all other particles.
Derivations of such mean field equations from the microscopic $N$ body Schr\"odinger evolution are usually done
for the reduced one particle density and are naturally based on hierarchies \cite{Spohn,erdos,elgart,erdos1,erdos2, adami}.

In the recent years alternative methods have been succesfully used to derive the Hartree equation from microscopic dynamics.
One approach was developed by Fr\"ohlich et al. using dispersive estimates and counting of Feynman graphs \cite{froehlich}. Another one was introduced by Rodnianski and Schlein \cite{rodnianskischlein}. They focus on the dynamics of coherent states, inspired by a semiclassical argument of Hepp \cite{hepp}.

We present here a new method for deriving mean field descriptions which is in particular when deriving the Gross-Pitaevskii equation simpler and more effective as it
yields more general results with greater ease. For concreteness of the presentation we consider a simple model leading to the Hartree equation.
 We consider
 a Bose condensate
of $N$ interacting particles when the external trap --- described by
an external potential $A^t$ --- is time varying, it can for example be removed.
We are interested in solutions of the $N$-particle Schr\"odinger
equation in units $\hbar=m=1$ \be\label{schroe}
i\dot \Psi_N^t = H\Psi_N^t \ee with symmetric initial wave function $\Psi_N^0$ when $N$ gets large and the interaction gets weak with $N$: The  Hamiltonian \be\label{hamiltonian}
 H_N=-\sum_{j=1}^N \Delta_j+\sum_{1\leq j<k\leq N} v_N^\beta(x_j-x_k) +\sum _{j=1}^N A^t(x_j)
 \ee
acts on the Hilbert space $\LZN$, and $\beta\in\mathbb{R}$
determines the scaling behavior of the interaction. Usually  $v_N^\beta$
scales with the particle number such  that the total
interaction energy scales in the same way as the total kinetic
energy of the $N$ particles. This means that the $L^1$-norm of  $v_N^\beta$ is proportional to $N^{-1}$, for example
$v_N^\beta(x)= N^{-1+3\beta} v(N^\beta x)$ for a compactly
supported, spherically symmetric, positive potential $v\in
L^\infty$. Thus the total interaction energy is for sufficiently smooth wave functions $\Psi$ of order $N$. 
For positive $\beta$ the support of the potential shrinks with $N$.
 As long as $\beta<1/3$ the interaction potentials overlap and the mean field approximation is heuristically clear.

For $1/3\leq\beta\leq 1$ the interactions get more $\delta$-like and do not overlap. But in this case the wave function $\Psi$ develops
 on the scale of the potential a structure around the centers of the interactions to keep the energy low. 
If the energy of $\Psi$ is controllable  the interaction effectively still behaves like a smeared out interaction 
with moderate scaling behavior and the mean field argument still holds \cite{pickl2}.

The trap potential $A^t$ 
does not depend on $N$.  $H_N$ conserves symmetry, i.e. 
any symmetric function $\Psi_N^0$ evolves into a symmetric function
$\Psi_N^t$.

Assume for the sake of simplicity for now that  the initial wave functions $\Psi_N^0$ is a product state $\Psi_N^0=\prod_{j=1}^N\phi^0(x_j)$ 
where $\phi^0\in L^2$. 

In the mean field limit the product structure survives during the time evolution, i.e.
for $N$ large $\Psi_N^t\approx\prod_{j=1}^N\phi^t(x_j)$, but every particle moves in the average field of all other particles (mean field) so that
$\phi^t$  solves the Schr\"odinger
equation \be\label{meanfield}i\dot
\phi^t=\left(-\Delta+A^t+V^\beta_{\phi^t}\right)\phi^t\ee (
with $\phi^0$ as above).
The ``mean field'' $V^\beta_{\phi^t}$
depends on $\phi^t$ itself, so (\ref{meanfield}) is a non-linear
equation.

Our new strategy revolves around the meaning of $\Psi_N^t\approx\prod_{j=1}^N\phi^t(x_j)$  for $N$ large:
``Most'' particles behave in a good way and the term on the right  has ``mostly'' product structure while only ``few'' particles will behave badly and 
will become entangled. We shall therefore introduce a biased counting of good and bad particles yielding a counting measure $\alpha(t)$ such that for $\alpha\approx 0$ most particles are good.   The algorithm will then produce an equation for $\alpha(t)$ which shows
that if $\alpha(0) \approx 0$ then  $\alpha(t) \approx 0$.
This result is then easily generalized to non product initial states and also to an assertion about the reduced  one particle density matrix, which is usually the way limits are phrased in \cite{erdos1,erdos2,rodnianskischlein,Spohn}:
$$\mu_1^{\Psi_N^t}(x,y):=\int
\Psi_N^{t*}(x,\ldots,x_N)\Psi_N(y,x_2,\ldots,x_N)d^3x_2\ldots
d^3x_N\;.$$ $\mu_1^{\Psi_N^t}$ converges to $|\phi^t\rangle\langle\phi^t|$ in trace norm. Such results are usually based
on a hierarchical method  analogous to BBGKY hierarchies.

A warning: One may be inclined to think of $\Psi_N^t\approx\prod_{j=1}^N\phi^t(x_j)$ in $L^2$ sense.
That is however false: Assume that
\be\label{example}\Psi_N=\prod_{j=1}^{N-1}\phi(x_j)\phi^\bot(x_N)\;\;\;\text{ for some }
\phi^\bot\bot\phi\;.\ee Then of course
$\Psi_N\bot\prod_{j=1}^N\phi(x_j)$, so in $L^2$-sense $\Psi_N$ is
far away from $\prod\phi$. 
Nevertheless clearly ``most'' 
 particles of $\Psi_N$ are in the state $\phi$. That is what we must focus on.

\section{The strategy}

We wish to control the number of bad particles in the condensate. The idea is simply to count (but in a biased way)  the relative number of particles
not in the state $\phi^t$. This then  leads to the following definition:

\begin{definition}
For any $\phi\in\LZ$ and any $1\leq j\leq N$ the
projectors $p_j^\phi:\LZN\to\LZN$ and $q_j^\phi:\LZN\to\LZN$ are given by
$$p_j^\phi\Psi_N=\phi(x_j)\int\phi^*(x_j)\Psi_N(x_1,\ldots,x_N)d^3x_j\;\;\;\forall\;\Psi_N\in\LZN$$
and $q_j^\phi=1-p_j^\phi$. We shall also use the bra-ket notation
$p_j^\phi=|\phi(x_j)\rangle\langle\phi(x_j)|$.

Furthermore we define on $\LZN$ for any $0\leq k\leq N$ the projector
$$P_{N,k}^\phi:=\left(\prod_{j=1}^kq_j^\phi\prod_{j=k+1}^Np_j^\phi\right)_{\text{sym}}\;.$$
The index sym means the following: For any $0\leq k\leq j\leq N$ consider the set
$$\mathcal{A}_k:=\{(a_1,a_2,\ldots,a_N): a_j\in\{0,1\}\;;\;
\sum_{j=1}^N a_j=k\}\;.$$  Then
$$P_{N,k}^\phi:=\sum_{a\in\mathcal{A}_k}\prod_{j=1}^N\big(p_{j}^{\phi}\big)^{1-a_j} \big(q_{j}^{\phi}\big)^{a_j}\;.$$
\end{definition}

With $P_{N,k}^\phi$ at hand we can define an object which ``counts
the number of particles which are not in the state $\phi^t$''.

\begin{definition}
Let $\laa\cdot,\cdot\raa$ be the scalar product on $\LZN$. We define for any $N\in\mathbb{N}$ and any function
$n(k):\{0,1,\ldots,N\}\to\mathbb{R}^+_0$ the functional
$\alpha_N:\LZN\times\LZ\to\mathbb{R}^+_0$ as
$$\alpha_N(\Psi_N,\phi):=\laa\Psi_N,\sum_{k=0}^N n(k)P_{N,k}^\phi\Psi_N\raa\;.$$

\end{definition}
Note that $\cup_{k=0}^N\mathcal{A}_k=\{0,1\}^N$, hence
$\sum_{k=0}^N P_{N,k}^\phi\equiv1$.

Let us explain a bit more what $\alpha_N$ does. Choose $n(k)=k/N$. Then  that part of
$\Psi_N$ where $k$ of the $N$ particles are not in the state
$\phi^t$, (i.e. $\laa\Psi_N, P_{N,k}^\phi\Psi_N\raa$) is given
the weight $k/N$. Hence 
$\alpha_N$ ``counts the relative number of particles which are not in the
state $\phi^t$''. 

The important role $\alpha_N$ plays lies however in the fact that the function $n(k)$ can be chosen more appropriately depending on the 
particular problem (see section \ref{weight} below). The choice will be dictated by the requirement that
$\alpha_N(\Psi_N,\phi)\to0$
implies convergence of $\mu_1^{\Psi_N}$ to
$|\phi\rangle\langle\phi|$ in trace norm.

\subsection{Control of $\alpha$}

Let us use the shorthand notation $\alpha(t):=\alpha_N(\Psi_N^t,\phi^t)$.
Our goal is to prove that $\lim_{N\to\infty}\alpha(t)=0$ if $\lim_{N\to\infty}\alpha(0)=0$. 
Using  Gr\o nwall it is sufficient to show that $$|\dot\alpha(t)|\leq C \alpha(t)+\landau(1)\;.$$
Now
\be\label{dtalphaoben}\dot{\alpha}_N(\Psi_N^t,\phi^t):=-i
\laa\Psi_N^t,[H_N-H^{mf}_N,\sum_{k=0}^N n(k)P_{N,k}^\phi]\Psi_N^t\raa\ee
where in view of (\ref{meanfield}) $H^{mf}_N$ is the $N$-body mean field Hamiltonian $H^{mf}_N=\sum_{j=1}^N-\Delta_j+A^t(x_j)+V^\beta_{\phi^t}(x_j)$.

Note that many of the terms appearing in $H_N$ (c.f. (\ref{hamiltonian})) and $H^{mf}_N$ cancel in the difference $H_N-H^{mf}_N$. All that remains is the 
interaction potential minus the mean field potential.

It is important to note that in this algorithm no propagation estimates on $\Psi_N^t$ are needed. It is possible to 
estimate the right hand side of (\ref{dtalphaoben}) $$\laa\Psi_N,[H_N-H^{mf}_N,\sum_{k=0}^N n(k)P_{N,k}^\phi]\Psi_N\raa$$
in terms of $\alpha(\Psi_N,\phi)$, $\landau(1)$ and the energy of $\Psi_N$ {\it uniformly} in $\Psi_N$ and $\phi$.

\subsection{Advantages of the method}\label{adv}

\begin{itemize}
 \item No hierarchies appear

\item No propagation estimates are needed

\item
The freedom of choice of the weight $n(k)$ gives flexibility in the proof (see section \ref{weight} below).

\end{itemize}

\subsection{Convergence of the Reduced Density Matrix}

An important feature of the functionals $\alpha_N$ defined above is that
\linebreak $\lim_{N\to\infty}\alpha_N(\Psi_N,\phi)=0$ implies convergence of
the reduced one particle density matrix for many different weights $n(k)$ and vice versus. 
 In particular we can state both our condition ($\alpha(0)\to0$) and our result ($\alpha(t)\to0$) on the respective convergence of the  
reduced one particle density matrix instead.

Before we prove this equivalence note first that the
weights $n(k)=k/N$ has the special property that
\be\label{easy}\widehat{n}^{\phi}:=\sum_{k=0}^N\frac{k}{N}P_{N,k}^{\phi}=N^{-1}\sum_{j=1}^Nq_j^\phi\;.\ee

To see this we use that $\sum_{k=0}^NP_{N,k}^\phi=1$, hence it
suffices to show that for any $0\leq k\leq N$
$$P_{N,k}^\phi\sum_{j=0}^N\frac{j}{N}P_{N,j}^{\phi}=P_{N,k}^\phi N^{-1}\sum_{j=1}^Nq_j^\phi\;.$$
The left side equals $\frac{k}{N}P_{N,k}^\phi$. Recall that
$P_{N,k}^\phi=(\prod_{j=1}^kq_j^\phi\prod_{j=k+1}^Np_k^\phi)_{sym}$.
Multiplying this with $\sum_{j=1}^Nq_j^\phi$ yields
a factor $k$. Thus also the right hand side equals
$\frac{k}{N}P_{N,k}^\phi$.

Because of (\ref{easy}) among the different weights $n(k)=k/N$ is easiest to handle. Therefore we shall show first that $\lim_{N\to\infty}\laa\Psi_N,\widehat{n}^\phi\Psi_N\raa=0$
is equivalent to convergence of the reduced density matrix and generalize to other weights thereafter.

Note that convergence of $\mu_1^{\Psi_N}$ to $|\phi\rangle\langle\phi|$ in trace norm is equivalent to convergence in operator norm, since $|\phi\rangle\langle\phi|$ is a rank
one projection \cite{rodnianskischlein}. Other equivalent definitions of asymptotic 100\% condensation can be found in \cite{michelangeli}.

\begin{lemma}\label{kondensat}
Let $j>0$ and $n(k)=k/N$. Let $\phi\in L^2$ and let $\Psi_N\in\LZN$ be symmetric. Let
$\mu_1^{\Psi_N}$ be the reduced one particle density matrix of
$\Psi_N$. Then
\begin{enumerate}

\item $$\lim_{N\to\infty}\laa\Psi_N,\widehat{n}^\phi\Psi_N\raa=0
\;\;\;\Leftrightarrow\;\;\;
\lim_{N\to\infty}\mu_1^{\Psi_N}=|\phi\rangle\langle\phi|\text{ in
operator norm.}$$ 
\item
$$\lim_{N\to\infty}\laa\Psi_N,\widehat{n}^\phi\Psi_N\raa=0
\;\;\;\Leftrightarrow\;\;\;
\lim_{N\to\infty}\laa\Psi_N,\left(\widehat{n}^\phi\right)^j\Psi_N\raa=0
\;.$$

\end{enumerate}
\end{lemma}

\begin{proof}
Note first that using (\ref{easy}) and symmetry of $\Psi_N$
$$\laa\Psi_N,\widehat{n}^\phi\Psi_N\raa=\|q_1^\phi\Psi_N\|^2\;.$$

\begin{itemize}

              \item[(a) ``$\Rightarrow$'']

Let
$$\lim_{N\to\infty}\laa\Psi_N,\widehat{n}^\phi\Psi_N\raa=0\;,$$ i.e. $\lim_{N\to\infty} \|q_1^\phi\Psi_N\|=0$ and $\lim_{N\to\infty}
\|p_1^\phi\Psi_N\|=1$. Note that \bea\label{mu} \mu_1^{\Psi_N}(x,y)&=&\int
\Psi_N(x,x_2,\ldots,x_N)\Psi_N^*(y,x_2,\ldots,x_N)d^{3N-3}x
\nonumber\\&=&\int
p_1^\phi\Psi_N(x,x_2,\ldots,x_N)p_1^\phi\Psi_N^*(y,x_2,\ldots,x_N)d^{3N-3}x
\\\nonumber&&+\int
q_1^\phi\Psi_N(x,x_2,\ldots,x_N)p_1^\phi\Psi_N^*(y,x_2,\ldots,x_N)d^{3N-3}x
\\\nonumber&&+\int
p_1^\phi\Psi_N(x,x_2,\ldots,x_N)q_1^\phi\Psi_N^*(y,x_2,\ldots,x_N)d^{3N-3}x
\\\nonumber&&+\int
q_1^\phi\Psi_N(x,x_2,\ldots,x_N)q_1^\phi\Psi_N^*(y,x_2,\ldots,x_N)d^{3N-3}x
\eea The first summand equals $\|p_1^\phi\Psi_N\|^2\;
|\phi\rangle\langle\phi|$. The second and third have operator norm
$\|q_1^\phi\Psi_N\|\;\|p_1^\phi\Psi_N\|$ and the fourth has operator norm $\|q_1^\phi\Psi_N\|^2$ and hence go to zero.

\item[(a) ``$\Leftarrow$'']

Assume that $\mu_1^{\Psi_N}\to|\phi\rangle\langle\phi|$ in operator norm as $N\to\infty$. It follows that
 $\lim_{N\to\infty}\langle\phi,\mu_1^{\Psi_N}\phi\rangle=1$.
Writing $\mu_1^{\Psi_N}$ like in  (\ref{mu}) and using that $q_1^\phi\phi(x_1)=0$ the second, third and fourth summand are zero. Therefore
\beas
\lim_{N\to\infty}\|p_1^\phi\Psi_N\|^2=1\;.
\eeas
Using now $p_1^\phi+q_1^\phi=1$ it follows that $\lim_{N\to\infty}\|q_1^\phi\Psi_N\|=0$.
\item[(b)]
For (b) we show that
\be\label{alter}\lim_{N\to\infty}\laa\Psi_N,\left(\widehat{n}^\phi\right)^j\Psi_N\raa=0
\Rightarrow
\lim_{N\to\infty}\laa\Psi_N,\left(\widehat{n}^\phi\right)^l\Psi_N\raa=0
\ee for any $j,l>0$, which is equivalent to (b).

Let
$\lim_{N\to\infty}\laa\Psi_N,(\widehat{n}^\phi)^j\Psi_N\raa=0$
for some $j>0$. We shall use the abbreviation
$$\delta_N:=\laa\Psi_N,(\widehat{n}^\phi)^j\Psi_N\raa=\sum_{k=0}^N
\left(\frac{k}{N}\right)^j\|P_{N,k}^\phi\Psi_N\|^2\;.$$ Let $k_N$ be
the smallest integer such that
$\left(\frac{k_N}{N}\right)^j<\sqrt{\delta_N}$. It follows that
$\left(\frac{k_N+1}{N}\right)^j\geq \sqrt{\delta_N}$ and thus
$\sum_{k_N+1}^N \|P_{N,k}^\phi\Psi_N\|^2\leq\sqrt{\delta_N}$. Hence
\beas\sum_{k=0}^N
\left(\frac{k}{N}\right)^l\|P_{N,k}^\phi\Psi_N\|^2&\leq&
\sum_{k=0}^{k_N}
\left(\frac{k}{N}\right)^l\|P_{N,k}^\phi\Psi_N\|^2+\sum_{k_N+1}^N
\|P_{N,k}^\phi\Psi_N\|^2
\\&\leq&\left(\frac{k_N}{N}\right)^l+\sqrt{\delta_N}
\leq \left(\sqrt{\delta_N}\right)^{l/j}+\sqrt{\delta_N}\;.\eeas Thus
$\lim_{N\to\infty}\laa\Psi_N,\left(\widehat{n}^\phi\right)^l\Psi_N\raa=0$
and (\ref{alter}) follows implying (b).

           \end{itemize}
\end{proof}

\begin{remark}\label{remimply}
Similarly one can proof that
$\lim_{N\to\infty}\laa\Psi_N,(\widehat{n}^\phi)^j\Psi_N\raa=0$
for $j\in\mathbb{R}^+$ implies convergence of the reduced
$k$-particle density matrix for any fixed $k<\infty$.

Also note that for any $m(k)\leq n(k)$
$$\lim_{N\to\infty}\laa\Psi_N,\sum_{k=0}^Nm(k)P_{N,k}^\phi\Psi_N\raa=0
\Rightarrow
\lim_{N\to\infty}\laa\Psi_N,\sum_{k=0}^Nn(k)P_{N,k}^\phi\Psi_N\raa=0
\;.$$ From this follows that $\lim_{N\to\infty}\alpha_N(\Psi_N,\phi)=0$
implies convergence of the reduced one particle density matrix for
any weight dominated by $(k/N)^j$ for some positive $j$.
\end{remark}

\subsection{The role of the weight}\label{weight}

An important feature of this new method is the freedom of choice for the weight $n(k)$. In the instructional example below we will explicitely deal with the simplest scaling behaviour which is $\beta=0$. In this case we get a Gr\"onwall-type estimate for $\alpha(\Psi_t,\phi_t)$ for many different weights $n(k)$, in particular  for any weight that can be written as $n(k)=(k/N)^j$ for some positive $j$. We shall choose $n(k)=k/N$ below, which is due to (\ref{easy}) the most convenient choice. 

But for other situations other choices for the weight are more appropriate: When dealing with scalings $\beta>1$ one can 
either \begin{enumerate}
        \item use high purity of the condensate to control $\frac{d}{dt}\alpha(\Psi_t,\phi_t)$ although the interaction gets $\delta$-like.
\item control some of the kinetic energy and use smoothness of $\Psi$.
       \end{enumerate}
Both ideas can be worked out, and fundamentally depend on the choice for the weight. Let us explain
\begin{enumerate}
\item is worked out in \cite{pickl2} where we have to restrict ourselves to scalings $\beta<1/6$. Still the result is interesting since it is so far the only derivation of the Gross Pitaevskii equation without positivity condition on the interaction. Here we use scalings of the form $$n(k)=\left\{ 
    \begin{array}{cc}
                 k/N^{\gamma} & \text{if }k\leq N^{\gamma}  \\
                 0 & \text{else}  
    \end{array} 
   \right.$$ with $0<\gamma<1$.
Depending on $\gamma$, $\alpha\to 0$ stands for a different purity of the condensate: For $\gamma=0$ $\alpha\to0$ is equivalent to $L^2$-convergence of $\Psi(x_1,\ldots,x_N)$ against the full product $\prod_{j=1}^N \phi(x_j)$, for $\gamma=1$ we have the ``old'' weight $n(k)=k/N$. 
For ``large'' $\beta$ one needs high purity of the condensate to derive the Gross-Pitaevskii equation: For $\beta$ close to $1/6$ $\gamma$ has to be chosen close to $0$.
        \item is worked out  in \cite{pickl1} and gives --- assuming positivity of the interaction --- good results for all scalings $0<\beta\leq 1$:
Doing the estimates one arrives for any weight $n(k)=(k/N)^j$ with positive $j$ roughly at the following formula $$\frac{d}{dt} \alpha(\Psi_t,\phi_t)\leq C_j(\alpha(\Psi_t,\phi_t)+\landau(1)+\|\nabla_1 q_1^\phi\Psi_t\|^2)\;.$$
Now  $\|\nabla_1 q_1^\phi\Psi_t\|^2$ can be controlled using conservation of energy and splitting up $\langle\Psi_t,H^{mf}\Psi_t\rangle$. It turns out that for weights $n(k)=(k/N)^j$ with $j\leq 1/2$ one can show that $\|\nabla_1 q_1^\phi\Psi_t\|^2\leq C \alpha(\Psi_t,\phi_t)$. Choosing for example the weight $n(k)=\sqrt{k/N}$ one arrives at a Gr\"onwall type estimate for $\alpha(\Psi_t,\phi_t)$.

A similar idea can be used in the Hartree-case (i.e. $\beta=0$) when dealing with interactions with strong singularities (see \cite{picklknowles}). \end{enumerate}

\section{A simple example: The Hartree equation}

For the rest of this paper we will apply the algorithm to the case $\beta=0$ for ease of presentation.  
Nevertheless our result here is already better than previous results \cite{elgart,erdosyau,froehlich,rodnianskischlein} in the sense that we
can handle stronger singularities than Coulomb in the interaction.

With more technical effort it is possible to treat, even stronger singularities in the interaction \cite{picklknowles}.
Other scalings ($\beta>0$) are treated in \cite{pickl1} and \cite{pickl2}.

For the scaling $\beta=0$ the mean field potential is simply the convolution
$v\star |\phi^t|^2$, i.e. (\ref{meanfield}) becomes the Hartree
equation
\be\label{hartreeeq}i\partial^t\phi^t=h^H\phi^t=:\big(-\Delta+A^t+v\star|\phi^t|^2\big)\phi^t\;.\ee

We choose for the weight $n(k)=k/N$. Assuming that 
$\alpha_N(\Psi_N^0,\phi^0)\to 0$ as $N\to\infty$ we  show
that $\alpha_N(\Psi_N^t,\phi^t)\to0$.

\begin{theorem}\label{Hartree}
Let $v\in L^{2r}$ for some $r\geq1$. Let  $A^t$
be a time dependent potential. Assume that for any
$N\in\mathbb{N}$ there exists a solution $\Psi_N^t$ of the
Schr\"odunger equation (\ref{schroe}) and a solution $\phi^t$ of the
Hartree equation (\ref{hartreeeq}) with $\|\phi^t\|_{2s}\leq\infty$
for $s=\frac{r}{r-1}$. Then $$\alpha_N(\Psi_N^t,\phi^t)\leq
e^{\int_{0}^t C^\tau d\tau}\alpha_N(\Psi_N^0,\phi^0)+(e^{\int_{0}^t
C^\tau d\tau}-1)N^{-1}\;,$$ where
$C^t:=10\|v\|_{2r}\|\phi^t\|_{2s}$\;.

\end{theorem}
\begin{remark}

\begin{enumerate}
\item
If $\|\phi^t\|_\infty<\infty$ for all $t< \infty$ we
can handle interactions $v\in L^2$.
The Theorem generalizes the results in the literature to time
dependent external fields.
\item
Assuming $\lim_{N\to
\infty}\alpha_N(\Psi_N^0,\phi^0)=0$ and $\int_{0}^\infty\|\phi^\tau\|_{2s}d\tau<\infty$ the Theorem implies that $\lim_{N\to
\infty}\alpha_N(\Psi_N^t,\phi^t)=0$
 uniform in $t<\infty$.
\item There is a lot of literature on solutions of nonlinear Schr\"odinger equations (see for example \cite{ginibre}) showing that
our assumptions on the solutions of the Hartree equation can be satisfied for many different setups.

\end{enumerate}
\end{remark}

As mentioned above we do not need any propagation estimates on
$\Psi_N^t$.
 To emphasize this we prove a stronger statement than the one in the Theorem. We define the  functional
 $\gamma_N:\LZ\otimes\LZN\to\mathbb{R}$ by

\be\label{alphastr}\gamma_N(\Psi_N,\phi):=-i
\laa\Psi_N,[H_N-H^H_N,\widehat{n}^{\phi}]\Psi_N\raa\ee
where
$$H_N^H:=\sum_{j=1}^N -\Delta_j+A^t(x_j)+(v\star |\phi^t|^2)(x_j)$$
is the sum of Hartree Hamiltonians for each particle. Recall the $\phi^t$ solves (\ref{meanfield}), so $\alpha_N$
and $\gamma_N$ are such that

\be\label{timeder}\dot\alpha_N(\Psi_N^t,\phi^t)=\gamma_N(\Psi_N^t,\phi^t)\;.\ee

$\gamma_N(\Psi_N,\phi)$ will now be estimated in terms of $\alpha_N(\Psi_N,\phi)$ {\it
uniformly} in $\Psi_N$ and $\phi$.
\begin{lemma}\label{lemmauni}
Let $v\in L^{2r}$ for some $r\geq 1$. Then
$$|\gamma_N(\Psi_N,\phi)|\leq
10C^\phi\left(\alpha_N(\Psi_N,\phi)+
N^{-1}\right)$$ for all $\Psi_N\in L^2$ and all $\phi\in L^{2s}$
with $s=\frac{r}{r-1}$and $C^\phi:=\|v\|_{2r}\|\phi\|_{2s}$.
\end{lemma}
From this Lemma the Theorem follows in view of (\ref{timeder}) using
Gr\o nwall: Writing $C^t:=10\|v\|_{2r}\|\phi^t\|_{2s}$ and
$f:=\alpha_N(\Psi_N^t,\phi^t)+ N^{-1}$ we have $|\dot f^t|\leq  C^t
f^t$. It follows that $f^t\leq e^{\int_{0}^tC^\tau d\tau} f^0$, i.e.
$$\alpha_N(\Psi_N^t,\phi^t)+ N^{-1}\leq e^{\int_{0}^tC^\tau
d\tau}\left(\alpha_N(\Psi_N^t,\phi^t)+ N^{-1}\right)$$ which proves the Theorem.

\begin{proof}[Proof of Lemma \ref{lemmauni}] 
Recall that $$H_N-H^H_N=\sum_{1\leq j<k\leq N}v_N(x_k-x_l)-\sum_{l=1}^N(v\star|\phi|^2)(x_l)$$ and
$\widehat{n}^{\phi}_N=N^{-1}\sum_{j=1}^Nq_j^\phi$.
 Using symmetry of $\Psi_N$  and $1=p_1^{\phi}+q_1^{\phi}$ it follows
\beas \gamma_N(\Psi_N,\phi)&=&-iN^{-1}\sum_{j=1}^N \laa
\Psi_N,[\sum_{l<k}v_N(x_k-x_l)-\sum_{l=1}^Nv\star|\phi|^2(x_l),q^{\phi}_j] \Psi_N\raa
\\&=& -i\laa
\Psi_N,[\sum_{k=2}^Nv_N(x_k-x_1)-\left(v\star |\phi|^2\right)(x_1),q^{\phi}_1]
\Psi_N\raa
%
%
%
%
%
%
%
\\&=& -i\laa \Psi_N,\big((N-1)v_N(x_2-x_1)-\left(v\star |\phi|^2\right)(x_1)\big)q^{\phi}_1\Psi_N\raa\\&&+i\laa
\Psi_N,q^{\phi}_1\big((N-1)v_N(x_2-x_1)-\left(v\star |\phi|^2\right)(x_1)\big)\Psi_N\raa
\\&=& -i\laa
\Psi_N,q^{\phi}_1\big((N-1)v_N(x_2-x_1)-\left(v\star |\phi|^2\right)(x_1)\big)q^{\phi}_1\Psi_N\raa\\&&-i\laa
\Psi_N,p^{\phi}_1\big((N-1)v_N(x_2-x_1)-\left(v\star |\phi|^2\right)(x_1)\big)q^{\phi}_1\Psi_N\raa\\&&+i\laa
\Psi_N,q^{\phi}_1\big((N-1)v_N(x_2-x_1)-\left(v\star |\phi|^2\right)(x_1)\big)q^{\phi}_1\Psi_N\raa\\&&+i\laa
\Psi_N,q^{\phi}_1\big((N-1)v_N(x_2-x_1)-\left(v\star |\phi|^2\right)(x_1)\big)p^{\phi}_1\Psi_N\raa
\\&=&2\Im\left(\laa
\Psi_N,p^{\phi}_1\big((N-1)v_N(x_2-x_1)-\left(v\star |\phi|^2\right)(x_1)\big)q^{\phi}_1\Psi_N\raa\right)\;.
 \eeas
With $1=p_2^{\phi}+q_2^{\phi}$ we have that
\bea\label{dreisum}
&=& 2\Im\left(\laa
\Psi_N,p^{\phi}_1p^{\phi}_2\big((N-1)v_N(x_2-x_1)-\left(v\star |\phi|^2\right)\big)(x_1)q^{\phi}_1p^{\phi}_2\Psi_N\raa\right)
\\&&+2\Im\left(\laa
\Psi_N,p^{\phi}_1p^{\phi}_2\big((N-1)v_N(x_2-x_1)-\left(v\star |\phi|^2\right)(x_1)\big)q^{\phi}_1q^{\phi}_2\Psi_N\raa\right)
\nonumber\\&&+2\Im\left(\laa
\Psi_N,p^{\phi}_1q^{\phi}_2\big((N-1)v_N(x_2-x_1)-\left(v\star |\phi|^2\right)(x_1)\big)q^{\phi}_1p^{\phi}_2\Psi_N\raa\right)
\nonumber\\&&+2\Im\left(\laa
\Psi_N,p^{\phi}_1q^{\phi}_2\big((N-1)v_N(x_2-x_1)-\left(v\star |\phi|^2\right)(x_1)\big)q^{\phi}_1q^{\phi}_2\Psi_N\raa\right)\nonumber\;. \eea
Since for any selfadjoint $A$ on $\LZN$ the operator
$p^{\phi}_1q^{\phi}_2Aq^{\phi}_1p^{\phi}_2$ is invariant under adjunction with simultaneous exchange
 of the variables $x_1$ and $x_2$, we see that the third
summand is zero.

Let under slight abuse of notation for any $j>0$ $$(\widehat{n}^\phi)^{-j}:=\sum_{k=1}^N \left(\frac{k}{N}\right)^{-j}P_{N,k}^\phi\:.$$
It follows that $(\widehat{n}^\phi)^{j}(\widehat{n}^\phi)^{-j}+P_{N,0}^\phi=1$, thus for any $1\leq j\leq N$
\be\label{nminus}(\widehat{n}^\phi)^{j}(\widehat{n}^\phi)^{-j}q^\phi_j=q^\phi_j\;.\ee
Defining $$V(x_1,x_2):=(N-1)v_N(x_2-x_1)-\left(v\star |\phi|^2\right)(x_1)$$ and using (\ref{nminus}) on the second summand of (\ref{dreisum}) we get
 \bea\label{summe2}
|\gamma_N(\Psi_N,\phi)|&\leq&  2\left|\laa
\Psi_N,p^{\phi}_1p^{\phi}_2 V(x_1,x_2)q^{\phi}_1p^{\phi}_2\Psi_N\raa\right|
\\\nonumber&&+2|\laa
\Psi_N,p^{\phi}_1p^{\phi}_2V(x_1,x_2)(\widehat{n}^{\phi})^{1/2}(\widehat{n}^{\phi})^{-1/2}q^{\phi}_1q^{\phi}_2\Psi_N\raa|
\\\nonumber&&+2|\laa \Psi_N,p^{\phi}_1q^{\phi}_2V(x_1,x_2)q^{\phi}_1q^{\phi}_2\Psi_N\raa|\;.
\eea

The first summand is the most important. It becomes small because the interaction is well approximated by the mean field potential. Recalling the
notation $p_2^\phi=|\phi(x_2)\rangle\langle\phi(x_2)|$ and the scaling behavior of $v_N=N^{-1}v$ it follows that \beas
&&\hspace{-0.5cm}p^{\phi}_2 V(x_1,x_2)p^{\phi}_2=(N-1)p^{\phi}_2v_N(x_2-x_1)p^{\phi}_2-p^{\phi}_2\left(v\star |\phi|^2\right)(x_1)p^{\phi}_2
\\&&\hspace{0.5cm}=(1-N^{-1})|\phi(x_2)\rangle\langle\phi(x_2)| v(x_2-x_1)|\phi(x_2)\rangle\langle\phi(x_2)| -p^{\phi}_2\left(v\star |\phi|^2\right)(x_1)
\\&&\hspace{0.5cm}=(1-N^{-1})|\phi(x_2)\rangle \left(v\star |\phi|^2\right)(x_1)\langle\phi(x_2)|-p^{\phi}_2\left(v\star |\phi|^2\right)(x_1)
\\&&\hspace{0.5cm}=-N^{-1}p^{\phi}_2\left(v\star |\phi|^2\right)(x_1)\;.
\eeas
Hence the first summand in
(\ref{summe2}) equals
\bea\label{weiter}&&\hspace{-1.2cm}2N^{-1}\left|\laa
p^{\phi}_1p^{\phi}_2\Psi_N,\left(v\star
|\phi|^2\right)(x_1)q^{\phi}_1p^{\phi}_2\Psi_N\raa\right|
\nonumber\\&\leq& 2N^{-1}\|\left( v\star
|\phi|^2\right)(x_1)p^{\phi}_1p^{\phi}_2\Psi_N\|\;\|q^{\phi}_1p^{\phi}_2\Psi_N\|
\nonumber\\&\leq&2N^{-1}\|
\left(v\star |\phi|^2\right)(x_1)p^{\phi}_1\Psi_N\|
\nonumber\\&=&2N^{-1}\left(\laa
\Psi_N,|\phi(x_1)\rangle\langle\phi(x_1)|(v\star
|\phi|^2)^2(x_1)|\phi(x_1)\rangle\langle\phi(x_1)|\Psi_N\raa\right)^{1/2}
\;.\eea With Young's inequality $$\|\left(v\star |\phi|^2\right)^2\|_r=\|v\star
|\phi|^2\|_{2r}^2\leq (\|\phi^2\|_1\|v\|_{2r})^2=\|v\|_{2r}^2\;.$$
Using H\"older inequality recalling that $\frac{1}{s}+\frac{1}{r}=1$ (\ref{weiter}) is bounded by \be\label{first}2N^{-1}(\|(v\star
|\phi|^2)^2\|_r\|\phi^2\|_s)^{1/2}\leq2N^{-1}(\|v\|_{2r}^2\|\phi\|_{2s}^2)^{1/2}=2N^{-1}C^\phi\;.\ee

Next we estimate the second summand of (\ref{summe2}).
Using $\|\phi\|=1$ and the scaling $v_N=N^{-1}v$ we get
$$\sup_{x_2\in\mathbb{R}^3}\|V(\cdot,x_2)\|_r\leq (N-1)\|v_N\|_r+\|v\star
|\phi|^2\|_r <\|v\|_r+\|v\|_r \|\phi\|^2=2\|v\|_r\;.
$$ Thus we have in operator norm
\be\label{aop}\|p^{\phi}_1V^2(x_1,x_2)p^{\phi}_1\|_{op}\leq\|\phi^2\|_s \sup_{x_2\in\mathbb{R}^3}\|V^2(\cdot,x_2)\|_r\leq4\|\phi\|_{2s}^2 \|v\|_{2r}^2\leq 4(C^\phi)^2\;.\ee
Going back to (\ref{summe2}) and using Schwarz inequality the second summand
there is bounded by \be\label{142}
 2\|(\widehat{n}^{\phi})^{1/2}V(x_1,x_2)
p^{\phi}_1p^{\phi}_2\Psi_N\|\;
\|(\widehat{n}^{\phi})^{-1/2}q^{\phi}_1q^{\phi}_2\Psi_N\| \;.\ee
Using symmetry \bea\label{142b}&&\hspace{-0.5cm} \|(\widehat{n}^{\phi})^{1/2}V(x_1,x_2)
p^{\phi}_1p^{\phi}_2\Psi_N\|^2\nonumber\\&&\hspace{0.5cm}=\laa\Psi_N,
p^{\phi}_1p^{\phi}_2V(x_1,x_2)\widehat{n}^{\phi} V(x_1,x_2)
p^{\phi}_1p^{\phi}_2\Psi_N\raa
\nonumber\\&&\hspace{0.5cm}=N^{-1}\sum_{j=1}^N \laa\Psi_N, p^{\phi}_1p^{\phi}_2V(x_1,x_2)q_j^{\phi}V(x_1,x_2)
p^{\phi}_1p^{\phi}_2\Psi_N\raa
\nonumber\\&&\hspace{0.5cm}=\frac{N-2}{N} \laa\Psi_N, p^{\phi}_1p^{\phi}_2V(x_1,x_2)q_3^{\phi}V(x_1,x_2)
p^{\phi}_1p^{\phi}_2\Psi_N\raa
\nonumber\\&&\hspace{0.8cm}+\frac{2}{N}\laa\Psi_N, p^{\phi}_1p^{\phi}_2V(x_1,x_2)q_1^{\phi}V(x_1,x_2)
p^{\phi}_1p^{\phi}_2\Psi_N\raa
\nonumber\\&&\hspace{0.5cm}\leq\|q_3^{\phi}\Psi_N\|^2\;\|p^{\phi}_1V^2(x_1,x_2)p^{\phi}_1\|_{op}+\frac{2}{N}\|p^{\phi}_1V^2(x_1,x_2)p^{\phi}_1\|_{op}
\nonumber\\&&\hspace{0.5cm}\leq4(C^\phi)^2 \alpha_N(\Psi_N,\phi)+\frac{8(C^\phi)^2}{N}\;.
\eea Using symmetry and (\ref{easy})  \beas\nonumber
N(N-1)\|(\widehat{n}^{\phi})^{-1/2}q^{\phi}_1q^{\phi}_2\Psi_N\|^2&=&N(N-1)\laa\Psi_N,(\widehat{n}^{\phi})^{-1}q^{\phi}_1q^{\phi}_2\Psi_N\raa
\nonumber\\&\leq&\sum_{j,k=1}^N\laa\Psi_N,(\widehat{n}^{\phi})^{-1}q^{\phi}_jq^{\phi}_k\Psi_N\raa
\nonumber\\&=&N^2\laa\Psi_N,(\widehat{n}^{\phi})^{-1}(\widehat{n}^{\phi})^{2}\Psi_N\raa=N^2\alpha_N\;.
\eeas
where we used (\ref{aop}) to get the last line.
This and (\ref{142b})  yield that  (\ref{142}) (i.e. the second summand in (\ref{summe2})) is for $N>1$ bounded by
\be\label{second}3C^\phi\sqrt{\alpha_N(\Psi_N,\phi)}(\alpha_N(\Psi_N,\phi)+N^{-1})^{1/2}\leq 6C^\phi(\alpha_N(\Psi_N,\phi)+N^{-1})\;.\ee

Using Schwarz inequality and (\ref{aop}) the third summand in (\ref{summe2}) is bounded
by \bea\nonumber\label{third}
2\|V(x_1,x_2)p^{\phi}_1q^{\phi}_2\Psi_N\|\;\|q^{\phi}_1q^{\phi}_2\Psi_N\|&\leq&2\left(\|p^{\phi}_1V^2(x_1,x_2)p^{\phi}_1\|_{op}\right)^{1/2}\|q^{\phi}_2\Psi_N\|^2\\&\leq&
4C^\phi \alpha_N(\Psi_N,\phi)\;. \eea 

The bounds (\ref{first}), (\ref{second}) and (\ref{third}) of the three summands in (\ref{summe2}) imply the Lemma.

\end{proof}

{\bf Acknowledgement}: The paper was influenced by many helpful comments of Detlef D\"urr and Martin Kolb. I wish to thank Antti Knowles for pointing out the value of the method in deriving the Hartree equation. Helpful discussions with Volker Bach, Jean-Bernard Bru, J\"urg Fr\"ohlich and Jakob Yngvason are gratefully acknowledged.

\end{document}